\documentclass[twocolumn,aps,american,pra,floatfix,nofootinbib,superscriptaddress,longbibliography,10pt]{revtex4-2}

\usepackage{amsfonts}
\usepackage{amsmath}
\usepackage{amssymb}
\usepackage{amsthm}
\usepackage{bbm}
\usepackage{bm}
\usepackage{mathrsfs}
\usepackage{MnSymbol}
\usepackage{accents}
\usepackage{braket}
\usepackage{multirow}
\usepackage[table]{xcolor}

\usepackage{epsfig}
\usepackage{graphics}
\usepackage{graphicx}
\usepackage{subfigure}
\usepackage{tikz}
\usetikzlibrary{arrows.meta}
\usetikzlibrary{quotes,fit}
\usetikzlibrary{datavisualization.formats.functions}
\usepackage[matrix,frame,arrow]{xy}
\usepackage[compat=0.6]{yquant}
\useyquantlanguage{groups}

\usepackage{times}

\usepackage{soul}

\usepackage[pdfstartview=FitH]{hyperref}
\usepackage{hypernat}
\hypersetup{
    colorlinks=true,   
    linkcolor=cyan,    
    citecolor=magenta, 
    filecolor=magenta, 
    urlcolor=cyan,     
    runcolor=cyan
}
\usepackage[capitalise]{cleveref} 

\newtheorem{theorem}{Theorem}
\newtheorem{lemma}{Lemma}
\newtheorem{corollary}{Corollary}
\newtheorem{remark}{Remark}

\newcommand{\bes}{\begin{subequations}}
\newcommand{\ees}{\end{subequations}}
\newcommand{\beq}{\begin{equation}}
\newcommand{\eeq}{\end{equation}}

\DeclareMathOperator{\Tr}{Tr}

\begin{document}

\author{Domenico D'Alessandro}
\affiliation{Department of Mathematics, Iowa State University, Ames, Iowa, USA}

\author{Phattharaporn Singkanipa}
\affiliation{Center for Quantum Information Science \& Technology, University of
Southern California, Los Angeles, CA 90089, USA}
\affiliation{Department of Physics \& Astronomy, University of Southern California,
Los Angeles, California 90089, USA}

\author{Daniel Lidar}
\affiliation{Center for Quantum Information Science \& Technology, University of
Southern California, Los Angeles, CA 90089, USA}
\affiliation{Department of Physics \& Astronomy, University of Southern California,
Los Angeles, California 90089, USA}
\affiliation{Department of Electrical \& Computer  Engineering,  University of Southern California, Los Angeles, California 90089, USA}
\affiliation{Department of Chemistry, University of Southern California,
Los Angeles, California 90089, USA}
\affiliation{Quantum Elements, Inc., Thousand Oaks, CA}

\title{Proof of the Error Scaling for Universally Robust Dynamical Decoupling Sequences}

\date{\today}

\begin{abstract}
Universally robust dynamical decoupling (UR$n$) sequences were proposed to compensate pulse imperfections arising from arbitrary experimental parameters while achieving high-order error suppression with only a linear increase in the number of pulses. Although their performance was supported by analytical arguments, numerical simulations, and experiments, a complete mathematical proof of the claimed order of error compensation has been absent. In this work, we present a rigorous proof for UR$n$ DD sequences with even $n$. Using a series expansion of a quantity whose modulus is the fidelity $F$, we derive necessary and sufficient conditions for the cancellation of its coefficients up to, but not including, order $n$. The UR$n$ phase prescription satisfies these conditions, and therefore $1-F=O(\epsilon^n)$. Our results establish the UR$n$ construction on firm analytical grounds and clarify the structure responsible for its high-order robustness.
\end{abstract}

\maketitle

\section{Introduction}

Protecting coherent quantum evolution against environmental noise and control imperfections is a central problem in quantum information processing, quantum sensing, magnetic resonance, and coherent optical storage. A basic strategy is dynamical decoupling (DD), beginning with Hahn spin echo~\cite{Hahn1950} and its multipulse extensions such as the Carr-Purcell~\cite{CarrPurcell1954}, and Carr-Purcell-Meiboom-Gill (CPMG)~\cite{MeiboomGill1958} sequences for nuclear magnetic resonance, and later generalized by the quantum information community into systematic open-loop decoupling methods and higher-order constructions~\cite{Viola:98,Viola1999,KhodjastehLidar2005,Santos:2006:150501,Uhrig2007,West:2010:130501,Wang:10,Xia:2011uq,Bookatz:2016aa,brown2024efficient,Yi:2026aa}. The common idea is to apply control pulses so that unwanted system-environment couplings are averaged out over the dynamical-evolution cycle~\cite{Viola:99a,Zanardi:1999fk,ByrdLidar:01}.

In realistic implementations, however, DD performance is often limited by pulse imperfections. This has motivated robust DD families designed to compensate control errors such as flip-angle and detuning errors, together with broader analyses of pulse-error tolerance in DD protocols~\cite{Souza2011,Souza2012,Quiroz:2013fv,SuterAlvarez2016}. In parallel, composite pulses provide a closely related route to robustness by replacing a single imperfect pulse with a phase-engineered sequence \cite{Levitt1986,Brown:04}. Universally robust DD sequences lie at this interface: they combine DD-type refocusing with phase engineering familiar from composite-pulse design.

Genov \textit{et al.} introduced universally robust DD sequences of order $n$ (UR$n$) for dephasing-dominated environments \cite{PhysRevLett.118.133202}. Their construction uses carefully chosen pulse phases so that systematic pulse errors cancel to increasingly high order, while the total number of pulses grows only linearly with the targeted compensation order. In the terminology of Ref.~\cite{PhysRevLett.118.133202}, ``universal'' refers to robustness with respect to pulse imperfections rather than to arbitrary noise models or computation. The analysis assumes identical phased pulses within one sequence and an environment whose correlation time is longer than the duration of that sequence. Closely related ideas were later developed for universally robust composite pulses with arbitrary excitation profiles~\cite{PhysRevA.101.013827}. 

Ref.~\cite{PhysRevLett.118.133202} supported the UR$n$ proposal by analytical arguments, numerical simulations, and experiments, but, to our knowledge, a complete mathematical proof of the claimed cancellation result for even $n$ has not appeared in the literature.
The purpose of the present paper is to supply that proof. We work in the same effective pulse-cycle model as \cite{PhysRevLett.118.133202}, expand the normalized Hilbert-Schmidt (HS) overlap $G(\epsilon):=\frac12\Tr(U_0^\dagger U_\epsilon)$ between the ideal and non-ideal unitary evolution operators $U_0$ and $U_\epsilon$ in the pulse-error parameter $\epsilon$, and derive exact necessary and sufficient conditions for the relevant amplitude coefficients to vanish. We then prove that the phase prescription of \cite{PhysRevLett.118.133202} satisfies those conditions for even $n$: all coefficients of $G(\epsilon)-1$ through order $n-1$ vanish for every value of the unknown effective phase $\alpha$, while the dependence on $\beta$ drops out of the HS overlap. Consequently, the fidelity $F=|G(\epsilon)|$ satisfies $1-F=O(\epsilon^n)$. The endpoint coefficient at order $\epsilon^n$ is not identically zero as a function of $\alpha$; hence, for fixed phases independent of $\alpha$, it is nonzero for generic $\alpha$.

This paper is organized as follows. In \Cref{sec:setup} we introduce the pulse model, explain the phase parametrization of imperfect pulses, and recall the UR phase prescription. In \cref{sec:math-form} we translate the robustness problem into a Taylor expansion of the HS overlap $G(\epsilon)$ and derive the trace identities whose verification implies high-order cancellation. This shows that proving high-order robustness is equivalent to proving certain Fourier-type identities, stated in \cref{subsec:main-th}, in particular \Cref{th:eq_40,th:eq_39}. The endpoint obstruction at order $\epsilon^n$ is handled separately in \cref{th:order-n1}. \Cref{sec:proof} contains the proofs of \Cref{th:eq_40,th:eq_39}. The appendices provide the proofs of auxiliary lemmas.

\section{Setup and definitions}\label{sec:setup}
A UR$n$ sequence consists of an even number $n$ of imperfect DD pulses separated by idle periods of free evolution. 
Each DD pulse is parametrized by the error parameter $\epsilon$ and unknown phases $\alpha$ and $\beta$, together with a controlled phase $\varphi$, as 
\beq\label{eq:Uephi}
U_\epsilon(\varphi):=\begin{pmatrix} \epsilon e^{i\alpha} & \sqrt{1-\epsilon^2} e^{i(\beta+\varphi)} 
\cr -\sqrt{1-\epsilon^2}e^{-i(\beta+\varphi)} & \epsilon e^{-i\alpha}
 \end{pmatrix}. 
\eeq

The physical interpretation of \cref{eq:Uephi} is discussed in \cite{PhysRevLett.118.133202}. In the URDD model, $U_\epsilon(\varphi)$ is an effective propagator for one free-evolution--pulse--free-evolution cycle. The parameters $p=1-\epsilon^2$, $\alpha$, and $\beta$ encode the net effect of pulse imperfections and accumulated phases during that cycle, while the experimentally controlled phase shift enters through the replacement $\beta\to\beta+\varphi$. We return to this effective-cycle picture in the next subsection.

The UR$n$ sequence is then given by a product of $n$ pulses of the form $U_\epsilon(\varphi_k)$, where the phases $\varphi_k$ ($k=1,2,\dots,n$) serve as control parameters, i.e.,
\beq
\label{eq:UE}
U_\epsilon:=U_\epsilon(\varphi_n)U_\epsilon(\varphi_{n-1})\cdots U_\epsilon(\varphi_1). 
\eeq
UR$n$ sequences are designed by choosing the phases $\{\varphi_k\}_{k=1}^n$ so as to maximize the fidelity
\beq\label{eq:fidelity1}
F:=\frac{1}{2}\left| \Tr(U_0^\dagger U_\epsilon )\right|,
\eeq
for arbitrary unknown $\alpha$ and $\beta$. Here $U_0$ denotes the propagator $U_\epsilon$ evaluated at $\epsilon = 0$, and $F$ is normalized so that $0\le F\le 1$. Equivalently, we write
\beq
G(\epsilon):=\frac12\Tr(U_0^\dagger U_\epsilon),
\quad F=|G(\epsilon)|.
\eeq
The trace is the HS inner product, defined more generally by
$\langle A,B\rangle := \Tr(A^\dagger B)$.
All coefficient cancellations below are proved for $G(\epsilon)$, or equivalently for the unnormalized HS overlap $\Tr(U_0^\dagger U_\epsilon)$. Note that this is stronger than what is needed for the fidelity. 

\subsection{Error and control model motivation}

The parametrization in \cref{eq:Uephi} is useful because it separates the experimentally controllable phase $\varphi$ from the unknown pulse imperfections. Following \cite{PhysRevLett.118.133202}, we model a single imperfect pulse by its transition probability $p=1-\epsilon^2$ together with two accumulated phases, here denoted by $\alpha$ and $\beta$, without assuming any particular pulse shape or microscopic error mechanism. In this way, variations in pulse amplitude, detuning, duration, and pulse shape, as well as slowly varying dephasing over one DD cycle, are absorbed into a small set of effective parameters. A controlled phase shift of the drive acts as $\Omega(t)\to \Omega(t)e^{i\varphi}$ and therefore enters the propagator through the replacement $\beta\to\beta+\varphi$. This makes the pulse phases $\varphi_k$ the natural control variables in UR sequences and explains why we can analyze robustness against arbitrary underlying experimental parameters at the level of the effective propagator. Throughout the UR construction, we further assume coherent evolution during a single DD sequence and identical phased pulses, so that the same effective pulse parameters apply from pulse to pulse \cite{PhysRevLett.118.133202}.
\subsection{Pulse decomposition and \texorpdfstring{$U_0$}{U0}}
We define a $z$-axis rotation $R(\varphi)$ and a modified $Y$ gate
\beq\label{eq:RZY}
R(\varphi):=\begin{pmatrix} e^{i\frac{\varphi}{2}} & 0 \cr 0 & e^{-i\frac{\varphi}{2}} \end{pmatrix}, \quad 
Y:=\begin{pmatrix}0 & e^{i(\beta -\alpha)}\cr - e^{-i(\beta -\alpha)} & 0 \end{pmatrix}, 
\eeq
and rewrite $U_\epsilon(\varphi)$ [\cref{eq:Uephi}] as
\beq\label{eq:withR}
U_\epsilon(\varphi)=R(\varphi) U_\epsilon(0)R^\dagger(\varphi)=R(\varphi) U_\epsilon(0)R(-\varphi).
\eeq
Likewise, $U_\epsilon(0)$ is
\beq\label{eq:UE3}
U_\epsilon(0)=R(2\alpha)\left( \epsilon {\bf 1} +\sqrt{1-\epsilon^2}Y \right),
\eeq
where ${\bf 1}$ is a $2\times 2$ identity matrix. Note that
\beq\label{eq:verif}
R(c)YR(b)Y=-R(c-b).
\eeq

With $\varphi_{n+1}:=0$, we define two angles as a function of $\varphi_k$ for $k=1,2,\dots,n$,
\beq\label{eq:DC}
\Delta_k:=\varphi_k-\varphi_{k+1}, \quad \chi_k:=\Delta_k+2\alpha.
\eeq
Substituting these definitions into \cref{eq:withR} and using \cref{eq:UE}, we obtain
\beq\label{eq:UE2}
U_\epsilon=R(\Delta_n) U_\epsilon(0)R(\Delta_{n-1}) U_\epsilon(0)\cdots R(\Delta_1) U_\epsilon(0)R^{\dagger}(\varphi_1). 
\eeq

When $\epsilon=0$, \cref{eq:UE2} can be written as 
\beq\label{eq:U01}
\begin{aligned}
U_0&=R(\chi_n)YR(\chi_{n-1})Y\cdots R(\chi_1)YR^\dagger(\varphi_1)\\
   &=(-1)^{{n/2}} R\left( \chi_n-\chi_{n-1}+\chi_{n-2}-\cdots -\chi_1\right)R^\dagger(\varphi_1)\\
&=(-1)^{{n/2}} R\left( \Delta_n-\Delta_{n-1}+\Delta_{n-2}-\cdots -\Delta_1\right)R^\dagger(\varphi_1),
\end{aligned}
\eeq
where we have used the identity in \cref{eq:verif} repeatedly to obtain the second line, and the fact that $n$ is even to obtain the third line. 

A common shift $\varphi_j\mapsto \varphi_j+\theta$ conjugates both $U_\epsilon$ and $U_0$ by the same $R(\theta)$, and therefore leaves $\Tr(U_0^\dagger U_\epsilon)$ invariant. We fix this phase gauge by setting $\varphi_1=0$, so that, henceforth, $R^\dagger(\varphi_1)={\bf 1}$. Defining
\beq
\label{eq:Gamma}
\Gamma :=\sum_{j=1}^n (-1)^{j-1} \Delta_j ,
\eeq
and using \cref{eq:U01}, this gives
\beq
\label{eq:U_0^dag}
U_0^\dagger =(-1)^{n/2}R\left(\Gamma\right),
\eeq
and hence
\beq
\label{eq:fidelity2}
F=\frac{1}{2} \left| \Tr\left[ R\left(\Gamma\right) U_\epsilon \right]\right|.
\eeq
This choice also coincides with the phase convention used in \cite{PhysRevLett.118.133202}, as stated next in \cref{th:MAIN}.

\subsection{The chosen angles \texorpdfstring{$\varphi_k$}{phik}}
\begin{theorem}[UR phase prescription of \cite{PhysRevLett.118.133202}]\label{th:MAIN}
Let $n\ge 4$ be even. Define
\beq
\Phi:=
\begin{cases}
\pm \dfrac{\pi}{m}, & n=4m,\\[8pt]
\pm \dfrac{2m\pi}{2m+1}, & n=4m+2.
\end{cases}
\eeq
Then the choice $\varphi_1=0$ and, for $k=2,3,\dots,n$,
\beq\label{eq:choice}
\varphi_k:=(k-1)\varphi_2 + \frac{(k-1)(k-2)}{2}\Phi,
\eeq
has the following properties:
\begin{enumerate}
\item The normalized HS overlap satisfies
\beq
G(\epsilon)=1+O(\epsilon^n)
\eeq
for arbitrary $\alpha$ and $\beta$. Hence $1-F=O(\epsilon^n)$. Equivalently, all coefficients of $G(\epsilon)-1$ of orders $1,\dots,n-1$ vanish; since only even powers occur, this is cancellation through order $n-2$.

\item For this UR phase family, the coefficient of order $\epsilon^n$ in $1-F$ is not identically zero as a function of $\alpha$.
Hence, for generic $\alpha$, the fidelity error is of exact order $\epsilon^n$.

\item This order is optimal at fixed sequence length $n$: no $\alpha$-independent choice of phases $\{\varphi_k\}_{k=1}^n$ can make all non-constant coefficients of $G(\epsilon)-1$ through order $\epsilon^n$ vanish identically as functions of $\alpha$.

\end{enumerate}
\end{theorem}

Our goal is to provide a rigorous mathematical proof of this theorem. For that purpose, it is convenient to write the phase differences $\Delta_k$ defined in \cref{eq:DC} in terms of the chosen phases $\varphi_k$ from \cref{eq:choice}. For $1\le k\le n-1$, direct calculation gives
\beq\label{eq:Deltak2}
\Delta_k=-\varphi_2-(k-1)\Phi \quad (1\le k\le n-1),
\eeq
while for the last interval we have
\beq\label{eq:Deltan}
\Delta_n=\varphi_n=(n-1)\varphi_2+\frac{(n-1)(n-2)}{2}\Phi.
\eeq

We first consider $U_0$. Evaluating the sum in \cref{eq:fidelity2} gives
\beq
\Gamma = -n\varphi_2-\frac{n(n-2)}{2}\Phi.
\eeq
For the UR choices of $\Phi$, the quantity $\frac{n(n-2)}{2}\Phi$ is an integer multiple of $4\pi$, and therefore
\beq
R\left(\Gamma\right)=R(-n\varphi_2).
\eeq
Equivalently, using \cref{eq:U01},
\beq
U_0=(-1)^{n/2}R(n\varphi_2).
\eeq

\section{Taylor expansion of the HS overlap}\label{sec:math-form}
We now analyze the Taylor coefficients of the HS overlap $G(\epsilon)$, equivalently of the unnormalized quantity $\Tr(U_0^\dagger U_\epsilon)$, in the error parameter $\epsilon$. This allows us to formulate mathematical conditions that are equivalent to the cancellation of the amplitude coefficients up to the desired order.

\subsection{Taylor expansion of \texorpdfstring{$U_\epsilon$}{Uepsilon}}\label{subsec:taylor}
Substituting $U_\epsilon(0)$ from \cref{eq:UE3} in \cref{eq:UE2} and denoting for brevity $R_j:=R(\chi_j)$, we have 
\beq
\begin{aligned}
U_\epsilon &=
R_n(\epsilon {\bf 1}+\sqrt{1-\epsilon^2} Y)
R_{n-1}(\epsilon {\bf 1}+\sqrt{1-\epsilon^2} Y)
\cdots \\
& \quad \cdots 
R_1(\epsilon {\bf 1}+\sqrt{1-\epsilon^2} Y)
=
\sum_{l=0}^n \epsilon^{n-l}(1-\epsilon^2)^{l/2}\tilde P_l,
\end{aligned}
\eeq
where $\tilde P_l$ is the sum of all products $A_n A_{n-1}\cdots A_1$ such that each $A_j$ is either $R_j$ or $R_jY$, and exactly $l$ of the factors are of the latter type (in particular, $l$ denotes the number of $Y$'s that appear in each term in the sum). This sum contains $\binom{n}{l}$ terms. 
For example, when $n=4$,
\beq\label{eq:ex-P2}
\begin{aligned}
    \tilde P_2&= R_4YR_3YR_2R_1+R_4YR_3R_2YR_1+R_4YR_3R_2R_1Y\\
    &\quad+R_4R_3YR_2YR_1+R_4R_3YR_2R_1Y+R_4R_3R_2YR_1Y. 
\end{aligned}
\eeq
Note from the first line of \cref{eq:U01} that $U_0 = \tilde P_n$, i.e., when each factor in the product is $R_jY$.
Therefore, we have
\beq\label{eq:Startpoi}
\Tr(U_0^\dagger U_\epsilon)=\sum_{l=0}^n \epsilon^{n-l} (1-\epsilon^2)^{l/2} \Tr (\tilde P_n^\dagger \tilde P_l).
\eeq

\begin{remark}\label{Impooss} 
$\tilde P_n=U_0=(-1)^{n/2}R(n\varphi_2)$ is a diagonal matrix,
whereas all terms in the sum defining $\tilde P_l$ are antidiagonal when $l$ is odd (recall that $n$ is even). Therefore, the contributions with odd $l$ vanish, and \cref{eq:Startpoi} can be rewritten as
\beq\label{eq:Startpoi2}
\Tr(U_0^\dagger U_\epsilon)=\sum_{\substack{l=0\\l \text{ even}}}^n \epsilon^{n-l} (1-\epsilon^2)^{l/2} \Tr (\tilde P_n^\dagger \tilde P_l).
\eeq
Thus the unnormalized HS overlap has only even powers of $\epsilon$; after division by $2$, the same is true for $G(\epsilon)$. 
\end{remark}

Since the sum is over even $l$, we can rewrite \cref{eq:Startpoi2} using the binomial formula as 
\beq
    \Tr(U_0^\dagger U_\epsilon)=\sum_{\substack{l=0\\l \text{ even}}}^n \sum_{k=0}^{l/2} \binom{l/2}{k}(-1)^k \epsilon^{2k+n-l} \Tr (\tilde P^\dagger_n \tilde P_l).
\eeq
Thus the trace-amplitude expansion is a finite polynomial in $\epsilon^2$.  
To find  explicit coefficients  for the even powers of $\epsilon$, we rearrange the above sum as follows
\beq
\begin{aligned}
    &\Tr(U_0^\dagger U_\epsilon)
=\sum_{\substack{l=0\\l \text{ even}}}^n \sum_{k=-l/2}^0  \binom{l/2}{k + l/2}(-1)^{k+l/2} \epsilon^{2k+n} \Tr (\tilde P^\dagger_n \tilde P_l)\\
&\quad=\sum_{k=-{n/2}}^0 \sum_{\substack{l=-2k\\l \text{ even}}}^n \binom{l/2}{k+l/2} (-1)^{k+l/2}\epsilon^{2k+n} \Tr(\tilde P_n^\dagger \tilde P_l)\\
&\quad=
\sum_{k=0}^{n/2} \left( \sum_{\substack{l=n-2k\\l \text{ even}}}^n 
\binom{l/2}{k - n/2 + l/2}
(-1)^{k-{n/2}+l/2}\Tr(\tilde P_n^\dagger \tilde P_l) \right) \epsilon^{2k},
\end{aligned}
\eeq 
where we shifted the index in the inner sum in the first equality. 
In the second equality, we swapped the order of summation, and in the last line, we performed an additional index shift in $k$. Therefore, the coefficient of $\epsilon^{2k}$ after the index shift, denoted by $a_k$, is
\beq\label{eq:ak}
\begin{aligned}
    a_k&=\sum_{\substack{l=n-2k\\l \text{ even}}}^n \binom{l/2}{k - n/2 + l/2} (-1)^{k-{n/2}+l/2}\Tr(\tilde P_n^\dagger \tilde P_l) \\
    &=\sum_{\substack{j=0\\j \text{ even}}}^{2k}
\binom{{(n-j)/2}}{k - {j/2}} (-1)^{k-{j/2}} \Tr(\tilde P_n^\dagger \tilde P_{n-j}), 
\end{aligned}
\eeq
and after index rescaling $j=2s$, we obtain 
\beq\label{eq:akFinal}
a_k=\sum_{s=0}^k \binom{n/2-s}{k-s} (-1)^{k-s}
\Tr(\tilde P_n^\dagger \tilde P_{n-2s}).
\eeq

From \cref{eq:akFinal}, we immediately obtain
\beq
a_0=\Tr(\tilde P_n^\dagger\tilde P_n)=2.
\eeq
Thus, if
\beq
\Tr(U_0^\dagger U_\epsilon)=\sum_{k=0}^{n/2}a_k\epsilon^{2k},
\eeq
then $G(\epsilon)=\frac12\sum_{k=0}^{n/2}a_k\epsilon^{2k}$, so $a_0=2$ gives $G(0)=1$ and hence $F(0)=1$. 

Next, we derive a necessary and sufficient condition on $\Tr(\tilde P_n^\dagger \tilde P_{n-2s})$ for $s = 1, \dots, k_{\max}$ under which $a_k=0$ for every $k=1,2,\dots,k_{\max}$, where $k_{\max}$ is any fixed integer with $1\le k_{\max}\le n/2$. Later we will apply this with $k_{\max}=n/2-1$.

\begin{theorem}\label{eq:equivalence}
Fix an integer $k_{\max}$ with $1\le k_{\max}\le n/2$. Then the coefficients $a_1,a_2,\dots,a_{k_{\max}}$ in \cref{eq:akFinal} all vanish if and only if, for every $k=0,1,2,\dots,k_{\max}$,
\beq\label{eq:cond5}
\Tr(\tilde P_n^\dagger \tilde P_{n-2k})=2\binom{n/2}{k}.
\eeq 
\end{theorem}

\begin{proof}
For $k=0$, \cref{eq:akFinal} gives
\beq
a_0=\Tr(\tilde P_n^\dagger \tilde P_n)=2,
\eeq
so  \cref{eq:cond5} is verified for $k=0$ .  Next  we  prove by induction on $k$ that $a_1=\cdots=a_k=0$ is equivalent to 
$\Tr(\tilde P_n^\dagger \tilde P_{n-2s})=2\binom{n/2}{s}$
for $s=1,\dots,k$.

Assume this equivalence has already been proved up to level $k-1$ (i.e., 
$a_1=\cdots=a_{k-1}=0$ is equivalent to
$\Tr(\tilde P_n^\dagger \tilde P_{n-2s})=2\binom{n/2}{s}$ for $s=1,\dots,k-1$).
Under this induction hypothesis, \cref{eq:akFinal} shows that the condition $a_k=0$ is equivalent to
\beq
\Tr \left(\tilde P_n^\dagger \tilde P_{n-2k}\right)
=
2\sum_{s=0}^{k-1}(-1)^{k-1-s}
\binom{n/2-s}{k-s}
\binom{n/2}{s}.
\eeq

For integers $b\ge k\ge 1$, we have
\beq
\begin{aligned}
\binom{b-s}{k-s}\binom{b}{s}
&=
\frac{(b-s)!}{(k-s)!(b-k)!} \frac{b!}{s!(b-s)!}\\
&=
\frac{b!}{k!(b-k)!} \frac{k!}{s!(k-s)!}
=
\binom{b}{k}\binom{k}{s}.
\end{aligned}
\eeq
Therefore
\beq
\sum_{s=0}^{k-1}(-1)^{k-1-s}
\binom{b-s}{k-s}
\binom{b}{s}
=
\binom{b}{k}
\sum_{s=0}^{k-1}(-1)^{k-1-s}\binom{k}{s}.
\eeq
By the binomial theorem,
\beq
0=(1-1)^k=\sum_{s=0}^{k}(-1)^{k-s}\binom{k}{s},
\eeq
so
\beq
\sum_{s=0}^{k-1}(-1)^{k-1-s}\binom{k}{s}=1.
\eeq
Taking $b=n/2$, we obtain
\beq
\Tr(\tilde P_n^\dagger \tilde P_{n-2k})=2\binom{n/2}{k},
\eeq
which is exactly \cref{eq:cond5} at level $k$. This completes the induction and proves the theorem.
\end{proof}

Recall that $l$ in $\tilde P_l$ denotes the number of $Y$ factors appearing in each term of the sum. We are interested in $\tilde P_l$ for even $l=2,4,\dots,n$. Equivalently, writing $l=n-2k$, we study $\tilde P_{n-2k}$, which is the form we use in \cref{eq:cond5}. 

In order to express the necessary and sufficient condition \eqref{eq:cond5} in a more explicit form, we define three auxiliary functions, ${\cal L}$, $S$, and $W$, on finite increasing sequences of integer indices.  

\subsection{Auxiliary functions \texorpdfstring{${\cal L}$, $S$, and $W$}{L, S, and W}}\label{subsec:L-func}

Define the alternating linear form $\mathcal{L}(X)$ on an ordered $n$-tuple $X=(x_1,x_2,\dots,x_n)$ by
\beq
\mathcal{L}(X)=\sum_{j=1}^n(-1)^{j-1}x_j=x_1-x_2+x_3-\dots-x_n .
\eeq

Throughout this section, a sequence $r=(r_1,r_2,\dots,r_{2k})$ means a strictly increasing $2k$-tuple of integers in the interval $[1,n]$. 
Given a sequence $r$, we define $\mathcal{L}_r(X)$ as a generalization of $\mathcal{L}(X)$, with coefficients taking values $\pm 1$ as follows. Starting from $+1$, the signs alternate up to $r_1$. At $r_1$, the sign of $x_{r_1}$ is repeated once, after which the signs alternate again up to $r_2$. At $r_2$, the sign of $x_{r_2}$ is repeated, and the alternation then continues. This procedure is repeated for all subsequent indices. For example,
\beq
\mathcal{L}_{(2,5)}(x_1,x_2,\dots,x_7)=x_1-x_2-x_3+x_4-x_5-x_6+x_7.
\eeq
Here $r = (2,5)$, hence the signs of $x_2$ and $x_5$ are repeated at $x_3$ and $x_6$, respectively, and the signs at all other positions alternate.

With these definitions, the function $\mathcal{L}_r(X)$ can also be defined recursively, in terms of $\mathcal{L}(X)$, by setting $r_0=0$ and $r_{2k+1}=n$ as 
\beq\label{eq:Altrecurs}
\mathcal{L}_r(X)=\sum_{j=0}^{2k} (-1)^{j+r_j} \mathcal{L}(x_{r_j+1},\dots,x_{r_{j+1}}).
\eeq
See \cref{app:Altrecurs} for the proof. Here, when $r_j=r_{j+1}$, the corresponding block is empty and its contribution is understood to be $0$.

Given a sequence $r$, we define two more  functions, the {\it signature} { $S=S(r)$} and {\it weighted signature}  { $W=W(r)$}, 
\beq
S(r) := \sum_{j=1}^{2k}(-1)^{j+r_j},
\quad
W(r) := \sum_{j=1}^{2k}(-1)^{j+r_j} r_j,
\label{eq:SWdef}
\eeq
which will play a central role in the proof. 
To express $\tilde P_{n-2k}$ in \cref{eq:cond5} in terms of $\mathcal L_r$, we first reverse the order of the $\chi_j$'s and $\Delta_j$'s in \cref{eq:DC}. Define, for $j=1,2,\dots,n$,
\beq\label{eq:DC2}
D_j:=\Delta_{n-j+1}, \quad C_j=\chi_{n-j+1}=D_j+2\alpha.
\eeq
We also define
\beq\label{eq:XCXDXN}
\begin{aligned}
X_C&:=(C_1,C_2,\dots,C_n), \quad X_D:=(D_1,D_2,\dots,D_n)\\
X_N&:=(1,2,\dots,n).
\end{aligned}
\eeq

\begin{lemma}\label{TTlemma}
For $k=1,2,\dots,n/2$, write
\beq
\sum_{r^<}:=\sum_{1\le r_1<r_2<\cdots<r_{2k}\le n}.
\eeq
Then
\beq\label{eq:PTbasic2}
\begin{aligned}
\tilde P_{n-2k}
&=(-1)^{\frac n2+k}\sum_{r^<}R\bigl(\mathcal L_r(X_C)\bigr)\\
&= (-1)^{\frac n2+k}\sum_{r^<}
R\bigl(2\alpha S(r)\bigr)R\bigl(\mathcal L_r(X_D)\bigr).
\end{aligned}
\eeq
\end{lemma}

\begin{proof}
After the reversal in \cref{eq:DC2}, write $\bar R_j:=R(C_j)$. From the definition of $\tilde P_l$ given at the beginning of \cref{subsec:taylor}, a term in $\tilde P_{n-2k}$ is obtained by choosing exactly $2k$ positions
\beq
1\le r_1<\cdots<r_{2k}\le n
\eeq
at which the factor is $\bar R_j$, while at all remaining positions the factor is $\bar R_jY$. Let $T_r$ denote the corresponding term, and write $r_0=0$, $r_{2k+1}=n$.

Group $T_r$ into the consecutive blocks determined by the tuple $r$. On the block from $r_j+1$ to $r_{j+1}$, repeated application of \cref{eq:verif}
shows that the contribution of that block to the total rotation angle is
$(-1)^{j+r_j} \mathcal L(C_{r_j+1},\dots,C_{r_{j+1}})$,
exactly as in the recursive definition \cref{eq:Altrecurs} of $\mathcal L_r(X_C)$. Since the total number of $Y$ factors is $n-2k$, the number of $Y$-pairs is $(n-2k)/2$, and each pair contributes one minus sign. Therefore
\beq
\begin{aligned}
T_r&=(-1)^{(n-2k)/2}R\bigl(\mathcal L_r(X_C)\bigr)\\
&= (-1)^{\frac n2+k}R\bigl(\mathcal L_r(X_C)\bigr).   
\end{aligned}
\eeq
Summing over all admissible $r^<$ gives the first identity in \cref{eq:PTbasic2}.

For the second identity, use $C_j=D_j+2\alpha$ and the linearity of $\mathcal L_r$:
\beq
\mathcal L_r(X_C)
=
\mathcal L_r(X_D)+2\alpha \mathcal L_r(1,\dots,1).
\eeq
It remains to compute $\mathcal L_r(1,\dots,1)$. By \cref{eq:Altrecurs},
\beq
\mathcal L_r(1,\dots,1)
=
\sum_{j=0}^{2k}(-1)^{j+r_j} \mathcal L(\underbrace{1,\dots,1}_{r_{j+1}-r_j}).
\eeq
Since
\beq
\mathcal L(\underbrace{1,\dots,1}_{m})=\frac{1-(-1)^m}{2},
\eeq
we obtain
\beq
\mathcal L_r(1,\dots,1)
=
\frac12\sum_{j=0}^{2k}(-1)^j\bigl((-1)^{r_j}-(-1)^{r_{j+1}}\bigr).
\eeq
Using $r_0=0$ and $r_{2k+1}=n$ with $n$ even, this simplifies to
\beq
\mathcal L_r(1,\dots,1)
=
\sum_{j=1}^{2k}(-1)^{j+r_j}
=
S(r).
\eeq
Therefore
\beq
\mathcal L_r(X_C)
=
2\alpha S(r)+\mathcal L_r(X_D),
\eeq
which proves the second identity in \cref{eq:PTbasic2}.
\end{proof}

The example of $\tilde P_2$ when $n=4$ in \cref{eq:ex-P2}, after reversing the order $\chi_j \to C_{n-j+1}$ and writing $\bar R_j:=R(C_j)$, can be written as
\beq\label{eq:ex-P2-inv}
\begin{aligned}
    \tilde P_2&= \bar{R}_1Y\bar{R}_2Y\bar{R}_3\bar{R}_4+\bar{R}_1Y\bar{R}_2\bar{R}_3Y\bar{R}_4+\bar{R}_1Y\bar{R}_2\bar{R}_3\bar{R}_4Y\\
    &\quad+\bar{R}_1\bar{R}_2Y\bar{R}_3Y\bar{R}_4+\bar{R}_1\bar{R}_2Y\bar{R}_3\bar{R}_4Y+\bar{R}_1\bar{R}_2\bar{R}_3Y\bar{R}_4Y. 
\end{aligned}
\eeq

Using \cref{TTlemma} in this example, the summation condition is $r^<:1 \le r_1 < r_2 \le n=4$. Accordingly, $r=(r_1,r_2)$ takes the values $(1,2)$, $(1,3)$, $(1,4)$, $(2,3)$, $(2,4)$, and $(3,4)$. These six choices correspond exactly to the six terms in \cref{eq:ex-P2-inv}, in agreement with the first line of \cref{eq:PTbasic2}.

Next, we relate $\mathcal{L}_r(X_N)$ to $S(r)$ and $W(r)$.
\begin{lemma}\label{Nicelemma}
Assume $n$ is even. For $r=(r_1,r_2,\dots,r_{2k})$,
\beq\label{eq:nicelemmaeq}
\mathcal{L}_r(X_N)=\frac{S(r)-n}{2}+W(r).
\eeq
\end{lemma}

The proof is given in \cref{app:Nicelemma}.

The next lemma computes $\mathcal{L}_{r}(X_D)$, which is needed in \cref{eq:PTbasic2}, explicitly for the UR phase choice of \cref{th:MAIN}.

\begin{lemma}\label{lem:LrXD}
With the definitions in \cref{eq:DC2,eq:Deltak2,eq:Deltan}, we have
\begin{equation}\label{eq:plu7}
\mathcal{L}_{r}(X_D)=n\varphi_2+\frac{n(n-1)}{2}\Phi-(\varphi_2+n\Phi)S(r)+\Phi \mathcal{L}_{r}(X_N).
\end{equation}
Equivalently, using \cref{eq:nicelemmaeq},
\begin{equation}\label{eq:LrXD-final}
\mathcal{L}_{r}(X_D)
=
n\varphi_2+\frac{n(n-2)}{2}\Phi
-\left(\varphi_2+n\Phi-\frac{\Phi}{2}\right)S(r)+\Phi W(r).
\end{equation}
\end{lemma}

\begin{proof}
From \cref{eq:DC2}, \cref{eq:Deltak2}, and \cref{eq:Deltan}, we obtain
\beq
\begin{aligned}
D_1&=(n-1)\varphi_2+\frac{(n-1)(n-2)}{2}\Phi,\\
D_j&=-\varphi_2-(n-j)\Phi=-(\varphi_2+n\Phi)+j\Phi
\quad (2\le j\le n).
\end{aligned}
\eeq
Write
\beq
\mathcal{L}_r(X)=\sum_{j=1}^n c_j x_j,
\quad c_j\in\{\pm1\}.
\eeq
By construction, the coefficient of the first entry is always $c_1=1$. Moreover,
\beq
\sum_{j=1}^n c_j=S(r),
\quad
\sum_{j=1}^n c_j j=\mathcal{L}_r(X_N).
\eeq
Hence
\beq
\begin{aligned}
\mathcal{L}_r(X_D)
&=D_1+\sum_{j=2}^n c_jD_j\\
&=D_1-(\varphi_2+n\Phi)\sum_{j=2}^n c_j+\Phi\sum_{j=2}^n c_j j\\
&=D_1-(\varphi_2+n\Phi)(S(r)-1)+\Phi(\mathcal{L}_r(X_N)-1).
\end{aligned}
\eeq
Substituting the value of $D_1$ and simplifying gives \cref{eq:plu7}. The second formula follows immediately by substituting \cref{eq:nicelemmaeq} into \cref{eq:plu7}.
\end{proof}

The preceding lemmas now combine to give the final form of $\tilde P_{n-2k}$ used in the rest of the paper.

\begin{corollary}\label{cor:PTbasic3}
Define
\beq\label{eq:eta}
\eta:=2 \alpha-\varphi_2-n\Phi+\frac{\Phi}{2}.
\eeq
Then, for the UR choices of $\Phi$,
\begin{equation}\label{eq:PTbasic3}
\tilde P_{n-2k}=(-1)^{\frac n2+k}\sum_{r^<}R\left(\eta S(r)\right)R\left(n\varphi_2+\Phi W(r)\right).
\end{equation}
\end{corollary}

\begin{proof}
Substituting \cref{eq:LrXD-final} into the second line of \cref{eq:PTbasic2} gives
\beq
\begin{aligned}
&\tilde P_{n-2k}
=(-1)^{\frac n2+k}\sum_{r^<}
R\bigl(2\alpha S(r)\bigr)\times\\
&\quad R \left(
n\varphi_2+\frac{n(n-2)}{2}\Phi
-\left(\varphi_2+n\Phi-\frac{\Phi}{2}\right)S(r)+\Phi W(r)
\right).
\end{aligned}
\eeq
Since all factors are $z$-axis rotations, they commute and their angles add. Therefore
\beq
\tilde P_{n-2k}
=
(-1)^{\frac n2+k}\sum_{r^<}
R \left(
\frac{n(n-2)}{2}\Phi+n\varphi_2+\eta S(r)+\Phi W(r)
\right).
\eeq
For the UR choices of $\Phi$, the quantity $\frac{n(n-2)}{2}\Phi$ is an integer multiple of $4\pi$, so $R \left(\frac{n(n-2)}{2}\Phi\right)=I$. Splitting the remaining angle into the $S(r)$-part and the $W(r)$-part yields \cref{eq:PTbasic3}.
\end{proof}

Since condition \cref{eq:cond5} must hold for arbitrary $\alpha$, it must equivalently hold for all values of $\eta$. This is the form of $\tilde P_{n-2k}$ that will be used in the proofs of the main theorems.

\section{Reduction to Fourier identities}\label{subsec:main-th}

For a UR$n$ sequence with even $n \ge 4$, recall the chosen value of $\Phi$ from \cref{th:MAIN} and define 
\beq
\omega := e^{i\Phi/2},
\quad
q := \omega^2 = e^{i\Phi}.
\eeq
In both UR families, and for either sign choice of $\Phi$, we have $\omega^n = 1$, and $q$ is a primitive $(n/2)$'th root of unity.

For a strictly increasing sequence $r$ of length $2k$, each summand in $S(r)$ is equal to $\pm1$. Since $2k$ is even, $S(r)$ is always even. Its possible values are therefore
\beq
0,\ \pm2,\ \pm4,\ \dots,\ \pm 2k.
\eeq

For an even integer $s$ and for an integer  $k$ satisfying $1 \le k \le n/2$, we also define
\beq
A_s^{(k)} := \sum_{\substack{1\le r_1<\cdots<r_{2k}\le n\\ S(r)=s}}
\omega^{W(r)}.
\label{eq:Asdef}
\eeq
Its complex conjugate is
\beq
\left(A_s^{(k)}\right)^*=\sum_{\substack{1\le r_1<\cdots<r_{2k}\le n\\ S(r)=s}}\omega^{-W(r)}.
\label{eq:Asconj}
\eeq

We now rewrite $\Tr(\tilde P_n^{\dagger}\tilde P_{n-2k})$ in \cref{eq:cond5} using \cref{eq:PTbasic3}. 
From \cref{eq:U01} and the evaluation of $\Gamma$, we have
\beq
\tilde P_n=U_0=(-1)^{n/2}R(n\varphi_2).
\eeq
Also, \cref{eq:PTbasic3} gives
\beq
\tilde P_{n-2k}=(-1)^{\frac n2+k}\sum_{r^<}R\left(\eta S(r)\right)R\left(n\varphi_2+\Phi W(r)\right).
\eeq
Therefore
\beq
\begin{aligned}
\Tr(\tilde P_n^\dagger \tilde P_{n-2k})
&=
(-1)^k\sum_{r^<}
\Tr \bigl(
R(-n\varphi_2)
R\left(\eta S(r)\right)\\
&\quad\quad\quad\times R\left(n\varphi_2+\Phi W(r)\right)
\bigr)\\
&=
(-1)^k\sum_{r^<}\Tr \left(R\left(\eta S(r)+\Phi W(r)\right)\right)\\
&=
2(-1)^k\sum_{r^<}
\cos\left(\frac{\eta}{2}S(r)+\frac{\Phi}{2}W(r)\right).
\end{aligned}
\eeq
Splitting the sum according to $S(r)=0$ and $|S(r)|=M$ gives
\beq
\begin{aligned}
&\Tr(\tilde P_n^\dagger \tilde P_{n-2k})
=
2(-1)^k
\sum_{\substack{r^<\\S(r)=0}}
\cos\left(\frac{\Phi}{2}W(r)\right)\\
&\quad+
2(-1)^k
\sum_{M=2,4,\dots,2k}
\sum_{\substack{r^<\\|S(r)|=M}}
\cos\left(\frac{\eta}{2}S(r)+\frac{\Phi}{2}W(r)\right).
\end{aligned}
\eeq
Using \cref{eq:Asdef,eq:Asconj}, this becomes
\beq\label{eq:trace-As}
\begin{aligned}
&\Tr(\tilde P_n^\dagger \tilde P_{n-2k})
=
(-1)^k\left( A_0^{(k)}+(A_0^{(k)})^*\right)\\
&\quad+
(-1)^k\sum_{M=2,4,\dots,2k}
\biggl[
e^{i\eta M/2}\left(A_M^{(k)}+(A_{-M}^{(k)})^*\right)\\
&\quad\quad\quad\quad\quad\quad
+
e^{-i\eta M/2}\left((A_M^{(k)})^*+A_{-M}^{(k)}\right)
\biggr].
\end{aligned}
\eeq
For a fixed $k$ with $1\le k<n/2$, the condition \eqref{eq:cond5} must hold for all values of $\eta$. By \cref{eq:trace-As}, this is equivalent to
\beq
A_0^{(k)}+(A_0^{(k)})^*
=
2(-1)^k\binom{n/2}{k},
\eeq
and, for every even $M=2,4,\dots,2k$,
\beq
A_M^{(k)}+(A_{-M}^{(k)})^*=0.
\eeq
The following two theorems establish exactly these Fourier-type identities for the UR$n$ phases.

\begin{theorem}[Nonzero-signature identity]\label{th:eq_40}
For the UR choices of $\Phi$, for every even $n\ge 4$ and every $1\le k<n/2$,
\beq
A_s^{(k)}=- \left(A_{-s}^{(k)}\right)^*,
\eeq
for all even integers $s$ with $0<|s|<n$.
\end{theorem}

\begin{theorem}[Zero-signature identity]\label{th:eq_39}
For the UR choices of $\Phi$, for every even $n\ge 4$ and every $1\le k<n/2$,
\beq
A_0^{(k)}=\left(A_0^{(k)}\right)^*=(-1)^k\binom{n/2}{k}.
\eeq
\end{theorem}

Note that the theorems above are stated only for $n\ge 4$.
When $n=2$, the only possible value is $k=n/2=1$, so the only coefficient beyond the leading constant term comes from $\tilde P_0$.

Together, \Cref{th:eq_40,th:eq_39} imply \eqref{eq:cond5} for every $1\le k<n/2$. By \Cref{eq:equivalence}, this gives
\beq
a_1=a_2=\cdots=a_{n/2-1}=0,
\eeq
and therefore $G(\epsilon)=1+O(\epsilon^n)$. 

It remains to analyze the endpoint $k=n/2$, which determines the coefficient of $\epsilon^n$.

\begin{theorem}[Endpoint obstruction and coefficient]
\label{th:order-n1}

Let $n\ge2$ be even and let $\varphi_1=\varphi_{n+1}=0$. 
If the lower-order coefficients $a_1,\dots,a_{n/2-1}$ vanish, then, 
with $\Gamma$ defined in \cref{eq:Gamma},
\beq
a_{n/2}
=
2(-1)^{n/2}\cos\left(n\alpha+\frac{\Gamma}{2}\right)-2.
\eeq
Consequently, $a_{n/2}$ cannot vanish identically as a function of $\alpha$. Therefore no choice of $n$ phases can make all nonconstant coefficients of $G(\epsilon)-1$ through order $\epsilon^n$ vanish for arbitrary $\alpha$.

For the UR phase family, $\Gamma\equiv -n\varphi_2\pmod{4\pi}$, and hence
\beq
a_{n/2}
=
2\left[
(-1)^{n/2}
\cos\left(n\alpha-\frac{n\varphi_2}{2}\right)-1
\right].
\eeq
Equivalently, for sufficiently small $\epsilon$,
\beq
\label{eq:1-F-expansion}
1-F
=
2\sin^2\left(
\frac{n\alpha}{2}
-\frac{n\varphi_2}{4}
+\frac{n\pi}{4}
\right)\epsilon^n.
\eeq
\end{theorem}

\begin{proof}
When $k=n/2$, we have $n-2k=0$, so $\tilde P_0$ is the unique term in the expansion with no $Y$-factors. Hence
\beq
\tilde P_0
=
R_nR_{n-1}\cdots R_1
=
R\left(\sum_{j=1}^n\chi_j\right).
\eeq
Using \cref{eq:DC},
\beq
\sum_{j=1}^n\chi_j
=
\sum_{j=1}^n(\Delta_j+2\alpha)
=
(\varphi_1-\varphi_{n+1})+2n\alpha
=
2n\alpha,
\eeq
because $\varphi_1=\varphi_{n+1}=0$. Thus
\beq
\tilde P_0=R(2n\alpha).
\eeq

By \cref{eq:U_0^dag},
\beq
\tilde P_n^\dagger=U_0^\dagger=(-1)^{n/2}R(\Gamma).
\eeq
Therefore
\beq
\begin{aligned}
\Tr(\tilde P_n^\dagger\tilde P_0)
&=
(-1)^{n/2}\Tr\left(R(\Gamma+2n\alpha)\right)\\
&=
2(-1)^{n/2}\cos\left(n\alpha+\frac{\Gamma}{2}\right).
\end{aligned}
\eeq

Assume now that $a_1,\dots,a_{n/2-1}$ vanish. By \cref{eq:equivalence}, this is equivalent to
\beq
\Tr(\tilde P_n^\dagger\tilde P_{n-2s})
=
2\binom{n/2}{s},
\qquad 0\le s<n/2.
\eeq
Using \cref{eq:akFinal} with $k=n/2$, we obtain
\beq
a_{n/2}
=
\sum_{s=0}^{n/2}
(-1)^{n/2-s}
\Tr(\tilde P_n^\dagger\tilde P_{n-2s}).
\eeq
The contribution from $0\le s<n/2$ is
\beq
2\sum_{s=0}^{n/2-1}
(-1)^{n/2-s}\binom{n/2}{s}
=
-2
\eeq
by the binomial theorem. Hence
\beq
a_{n/2}
=
\Tr(\tilde P_n^\dagger\tilde P_0)-2
=
2(-1)^{n/2}\cos\left(n\alpha+\frac{\Gamma}{2}\right)-2.
\eeq
This function is not identically zero as a function of $\alpha$, independently of the value of $\Gamma$.  Therefore the endpoint coefficient cannot be cancelled uniformly in $\alpha$.

For the UR phases, $\Gamma\equiv -n\varphi_2\pmod{4\pi}$, so
\beq
a_{n/2}
=
2\left[
(-1)^{n/2}
\cos\left(n\alpha-\frac{n\varphi_2}{2}\right)-1
\right].
\eeq
Since the expansion of the HS overlap contains no powers higher than $\epsilon^n$, the UR phase family gives
\beq
\begin{aligned}
G(\epsilon)
&=
1+
\left[
(-1)^{n/2}
\cos\left(n\alpha-\frac{n\varphi_2}{2}\right)-1
\right]\epsilon^n\\
&=1-2\sin^2\left(
\frac{n\alpha}{2}
-\frac{n\varphi_2}{4}
+\frac{n\pi}{4}
\right)\epsilon^n.
\end{aligned}
\eeq
Therefore
\beq
1-F
=
1-
\left|
1-
2\sin^2\left(
\frac{n\alpha}{2}
-\frac{n\varphi_2}{4}
+\frac{n\pi}{4}
\right)\epsilon^n
\right|.
\eeq
In particular, for sufficiently small $\epsilon$ we obtain \cref{eq:1-F-expansion}.
\end{proof}

To complete the proof of \cref{th:MAIN}, it remains to prove the Fourier identities of \Cref{th:eq_40,th:eq_39}. We do this in the following section.

\section{\texorpdfstring{Proofs of Theorems~\ref{th:eq_40} and \ref{th:eq_39}}{Proofs of the Fourier identities}}
\label{sec:proof}

We start by proving a property of the quantities $A_s^{(k)}$ defined in \cref{eq:Asdef}, which  will be used in the proofs of both theorems.

\begin{lemma}\label{lem:odd_reflection}
For every $1 \leq k < \frac{n}{2}$ and every even $s$ with $0\leq |s|
 \leq 2k$ 
\beq
A_s^{(k)}=\omega^s (A_s^{(k)})^{*}
\label{eq1}
\eeq
\end{lemma}
\begin{proof}
For a given $k$, on the set of increasing sequences
\beq\label{Dset}
\mathcal D := \Bigl\{  r=(r_1,\dots,r_{2k})\in\mathbb Z^{2k}\ :\ 1\le r_1<\cdots<r_{2k}\le n \Bigr\}, 
\eeq
we define a map   $\tau   :   \mathcal D \rightarrow  \mathcal D$ componentwise as 
\beq
\tau(r)_j  :=  n+1-r_{2k+1-j}\quad (1\le j\le 2k).
\eeq
Direct verification shows that $\tau$ does indeed preserve the ordering in a sequence $r$ and it is a bijection, in fact an involution ($\tau^2$ is the identity map), since 
\beq
\tau(\tau(r))_j
=n+1-\bigl(n+1-r_{2k+1-(2k+1-j)}\bigr)=r_j. 
\eeq
We  now verify how the signature and weighted signature of $r$ and $\tau(r)$ are related. Using the definition of $\tau$, we have 
\beq 
\begin{aligned}
(-1)^j (-1)^{\tau(r)_j}&=(-1)^j(-1)^{n+1-r_{2k-j+1}}\\
=(-1)^{j+1}(-1)^{r_{2k-j+1}}&
=(-1)^{2k-j+1}(-1)^{r_{2k-j+1}}, 
\end{aligned}
\eeq
since $2k+1$ is odd. Summing over $j=1,...,2k$ and using a change of summation index $l=2k-j+1$, we have that $S(\tau(r))=S(r)$. 

The weighted signature $W$ transforms as follows:

\beq
\begin{aligned}
   W(\tau(r))
&= \sum_{j=1}^{2k}    (-1)^j (-1)^{\tau(r)_j} \tau(r)_j\\
 &=\sum_{j=1}^{2k}    (-1)^j (-1)^{\tau(r)_j}(n+1-r_{2k-j+1})\\
 &=(n+1)S(r)-\sum_{j=1}^{2k}(-1)^j (-1)^{n+1-r_{2k-j+1}}r_{2k-j+1}\\
 &=(n+1)S(r)-\sum_{l=1}^{2k}(-1)^l (-1)^{r_{l}}r_{l}\\
 &=(n+1)S(r)-W(r)
\end{aligned}
\eeq
where we again used the change of index $l=2k-j+1$. Therefore, using $\omega^n=1$, we have 
\beq
\omega^{W(\tau(r))}
= \omega^{(n+1)S(r)}\omega^{-W(r)}
= \omega^{S(r)}\big(\omega^{W(r)}\big)^*.
\eeq
Summing over all tuples with $S(r)=s$ and changing variables $r\mapsto \tau(r)$, we obtain \cref{eq1}. 

\end{proof}

\subsection{Proof of \texorpdfstring{\cref{th:eq_40}}{the nonzero-signature identity}}

\begin{proof}
In order to simplify the notation, we omit the superscript $^{(k)}$ in $A_s^{(k)}$ and write simply $A_s$. If $|s|>2k$, then no admissible sequence of length $2k$ has signature $s$ or $-s$, so $A_s=A_{-s}=0$ and the claim is trivial. Hence we assume $0<|s|\le 2k$.

We first consider \cref{lem:odd_reflection} written for $s$, and then written for $-s$ and conjugated. Summing the two relations, we have
\beq\label{simrel}
A_s+A_{-s}^*=\omega^s(A_s+A_{-s}^*)^*.
\eeq
Let
\beq
B_s:=A_s+A_{-s}^*.
\eeq
From \cref{simrel}, if $B_s$ is real then
\beq
B_s=\omega^s B_s,
\eeq
so $B_s=0$ whenever $\omega^s\ne1$. For the UR choices of $\Phi$, no nonzero even integer $s$ with $|s|<n$ satisfies $\omega^s=1$. Indeed, if $n=4m$, then $\omega$ has order $n$. If $n=4m+2$, then $\omega$ has order $n$ when $m$ is odd and order $n/2=2m+1$ when $m$ is even. In the latter case $n/2$ is odd, so no nonzero even $s$ with $|s|<n$ can be a multiple of the order of $\omega$. Thus it remains only to prove that $B_s$ is real.

We decompose the set ${\cal D}$ in \cref{Dset} as the disjoint union of ${\cal D}_I$ and ${\cal D}_B$. Here $I$ stands for ``internal'', and ${\cal D}_I$ consists of all sequences with $r_{2k}<n$. The letter $B$ stands for ``boundary'', and ${\cal D}_B$ consists of all sequences with $r_{2k}=n$. Define
\beq\label{AsIAsB}
A_{s,I}:=\sum_{\substack{r \in {\cal D}_I\\ S(r)=s}} \omega^{W(r)},
\quad
A_{s,B}:=\sum_{\substack{r \in {\cal D}_B\\ S(r)=s}} \omega^{W(r)},
\eeq
so that
\beq\label{Adeco}
A_s=A_{s,I}+A_{s,B}.
\eeq
Similarly to the map $\tau$ used in \cref{lem:odd_reflection}, define
\beq
\tau_I:{\cal D}_I\to{\cal D}_I,
\quad
\tau_B:{\cal D}_B\to{\cal D}_B
\eeq
by
\beq\label{tauI}
\tau_I(r)_j=n-r_{2k+1-j},
\quad 1\le j\le 2k,
\eeq
and
\beq\label{tauB}
\tau_B(r)_j:=n-r_{2k-j}\quad(1\le j\le 2k-1),
\quad
\tau_B(r)_{2k}:=n.
\eeq
Direct verification similar to the one done for $\tau$ in \cref{lem:odd_reflection} shows that both $\tau_I^2$ and $\tau_B^2$ are the identity on the corresponding domains, and therefore the two maps are bijections on the corresponding domains. Furthermore, direct calculation gives
\begin{align}
S(\tau_I(r))&=-S(r),
&
W(\tau_I(r))&=W(r)-nS(r),
\label{SWI}\\
S(\tau_B(r))&=S(r),
&
W(\tau_B(r))&=n(S(r)+1)-W(r).
\label{SWB}
\end{align}
Since $\omega^n=1$, these identities imply
\beq
\omega^{W(\tau_I(r))}=\omega^{W(r)},
\quad
\omega^{W(\tau_B(r))}=\omega^{-W(r)}.
\eeq

Consider now $A_{s,I}$ in \cref{AsIAsB}. Using \cref{SWI}, we obtain
\beq\label{ASIprop}
A_{s,I}
=
\sum_{\substack{r \in {\cal D}_I\\ S(r)=s}} \omega^{W(r)}
=
\sum_{\substack{r \in {\cal D}_I\\ S(r)=s}} \omega^{W(\tau_I(r))}
=
\sum_{\substack{r \in {\cal D}_I\\ S(r)=-s}} \omega^{W(r)}
=
A_{-s,I}.
\eeq
Analogously, using \cref{SWB},
\beq\label{ASBprop}
A_{s,B}
=
\sum_{\substack{r \in {\cal D}_B\\ S(r)=s}} \omega^{W(r)}
=
\sum_{\substack{r \in {\cal D}_B\\ S(r)=s}} \omega^{W(\tau_B(r))}
=
\sum_{\substack{r \in {\cal D}_B\\ S(r)=s}} \omega^{-W(r)}
=
A_{s,B}^*.
\eeq
Thus $A_{s,B}$ is real. These relations are valid for every $s$, and therefore for $s$ and $-s$ simultaneously.

Using \cref{Adeco}, we have
\beq\label{finalfor}
A_s+A_{-s}^*
=
\left(A_{s,I}+A_{-s,I}^*\right)
+
\left(A_{s,B}+A_{-s,B}^*\right).
\eeq
The second parenthesis is real because both $A_{s,B}$ and $A_{-s,B}$ are real. The first parenthesis is real because \cref{ASIprop} gives $A_{s,I}=A_{-s,I}$. Hence { $B_s=A_s+A_{-s}^*$} is real, and the proof is complete.
\end{proof}

\subsection{Proof of \texorpdfstring{\cref{th:eq_39}}{the zero-signature identity}}

Fix $k$ with $1\le k<n/2$. We use a matrix-product generating function to prove the zero-signature identity. Replacing $\Phi$ by $-\Phi$ replaces $\omega$ by $\omega^{-1}$ and therefore conjugates each sum $A_s^{(k)}$. Since the claimed value of $A_0^{(k)}$ is real, it suffices to prove the result for the positive choice of $\Phi$.

For real $x$, define
\beq
\omega^x:=\exp(i\Phi x/2),
\quad
(\omega^x)^*=\omega^{-x}.
\eeq
We encode the sequences $1\le r_1<r_2<\cdots<r_l\le n$ of length $l$ in a noncommutative matrix product. This allows us to keep track of the different contributions of the elements of the sequence to the weighted signature $W$ and the signature $S$, according to the relative position of each term in the sequence. Specifically, we define a sequence $\{\beta_j\}_{j=1}^n$, chosen differently in the cases $n=4m$ and $n=4m+2$ with $m$ odd, and $n=4m+2$ with $m$ even, and set
\beq
K_j:=\begin{pmatrix}0&\omega^{\beta_j}\\ \omega^{-\beta_j}&0\end{pmatrix}.
\eeq
Define
\beq\label{Pt}
P(t):=(I+tK_1)(I+tK_2)\cdots(I+tK_n).
\eeq
We can write
\beq\label{DiagAdiag}
P(t)=
\sum_{\substack{0\le l\le n\\ l\ \mathrm{even}}}t^lD_l+
\sum_{\substack{0\le l\le n\\ l\ \mathrm{odd}}}t^lA_l.
\eeq
where $D_l$ is diagonal and $A_l$ is antidiagonal, say
\beq
D_l=\begin{pmatrix}\lambda_l&0\\0&\lambda_l^*\end{pmatrix},
\quad
A_l=\begin{pmatrix}0&\mu_l\\ \mu_l^*&0\end{pmatrix}.
\eeq

For real parameters $\gamma$ and $c$, define
\beq\label{betaj}
\beta_j:=(-1)^{j+1}\left(j+\gamma\right)+c.
\eeq
Then, for every sequence $r$ with $1\le r_1<r_2<\cdots<r_{2k}\le n$,
\beq
\begin{aligned}
\sum_{j=1}^{2k}(-1)^{j-1}\beta_{r_j}
&=
\sum_{j=1}^{2k}(-1)^j(-1)^{r_j}r_j
+
\gamma\sum_{j=1}^{2k}(-1)^j(-1)^{r_j} \\
&= W(r)+\gamma S(r).
\end{aligned}
\eeq
The constant $c$ drops out because
\beq
\sum_{j=1}^{2k}(-1)^{j-1}=0.
\eeq
Considering specifically the case of $l$ even in \cref{DiagAdiag}, with $l=2k>0$,
\beq\label{lambdal}
\lambda_{2k}=\sum_{1\le r_1< r_2<\cdots <r_{2k} \le n}\omega^{W(r)+\gamma S(r)},
\eeq
which we can write as
\beq\label{lambdal2}
\lambda_{2k}=A_0^{(k)} +\sum_{a=1}^k \left( \omega^{\gamma 2a} A_{2a}^{(k)} +\omega^{-\gamma 2a} A_{-2a}^{(k)} \right).
\eeq
By \Cref{th:eq_40}, for every nonzero even $s$ in the relevant range,
\beq
A_{-s}^{(k)}=-(A_s^{(k)})^*.
\eeq
Hence, for each $a\ge1$,
\beq
\omega^{2a\gamma} A_{2a}^{(k)}
+
\omega^{-2a\gamma} A_{-2a}^{(k)}
=
\omega^{2a\gamma} A_{2a}^{(k)}
-
\left(\omega^{2a\gamma} A_{2a}^{(k)}\right)^*,
\eeq
which is purely imaginary. On the other hand, \cref{lem:odd_reflection} with $s=0$ gives
\beq
A_0^{(k)}=(A_0^{(k)})^*,
\eeq
so $A_0^{(k)}$ is real. Therefore, once we prove below that
\beq
\lambda_{2k}=(-1)^k\binom{n/2}{k},
\eeq
comparison of real parts in \cref{lambdal2} gives
\beq
A_0^{(k)}=(-1)^k\binom{n/2}{k},
\eeq
which is the desired identity.

\subsubsection{Case \texorpdfstring{$n=4m$}{n=4m}}

Choose $\gamma=c=-\frac12$ in \cref{betaj}. Then
\beq\label{betajadd}
\beta_j=
\begin{cases}
j-1, & j\ \text{odd},\\
-j, & j\ \text{even}.
\end{cases}
\eeq
Write
\beq
P(t)=T_1(t)\cdots T_{2m}(t)T_{2m+1}(t)\cdots T_{4m}(t),
\eeq
where
\beq
T_j(t):=I+tK_j
=
\begin{pmatrix}
1 & t\omega^{\beta_j}\\
t\omega^{-\beta_j} & 1
\end{pmatrix}.
\eeq
Since $\omega^{\pm 2m}=-1$, \cref{betajadd} gives
\beq
\omega^{\beta_{j+2m}}=-\omega^{\beta_j},
\quad 1\le j\le 2m.
\eeq
Thus
\beq
T_{j+2m}(t)=E T_j(t) E,
\quad
E:=\begin{pmatrix}1&0\\0&-1\end{pmatrix}.
\eeq
If
\beq
Q(t):=T_1(t)\cdots T_{2m}(t),
\eeq
then
\beq
P(t)=Q(t)E Q(t)E.
\eeq

For $1\le j\le m$, direct verification from \cref{betajadd} gives
\beq
\beta_j-\beta_{2m-j+1}=\pm 2m.
\eeq
Hence
\beq
\omega^{\beta_{2m-j+1}}=-\omega^{\beta_j},
\quad
\omega^{-\beta_{2m-j+1}}=-\omega^{-\beta_j},
\eeq
and therefore
\beq
T_j(t)T_{2m-j+1}(t)=(1-t^2)I.
\eeq
Applying this first to $T_mT_{m+1}$, then to $T_{m-1}T_{m+2}$, and so on, gives
\beq
Q(t)=(1-t^2)^m I.
\eeq
Consequently,
\beq
P(t)=Q(t)E Q(t)E=(1-t^2)^{2m}I.
\eeq
Thus the coefficient $\lambda_{2k}$ of $t^{2k}$ in the $(1,1)$-entry of $P(t)$ is
\beq
\lambda_{2k}=(-1)^k\binom{2m}{k}=(-1)^k\binom{n/2}{k}.
\eeq
By the reduction preceding this case analysis, $A_0^{(k)}=\lambda_{2k}$. This proves \cref{th:eq_39} for $n=4m$.

We now turn to the case $n=4m+2$. Set
\beq
N:=\frac n2=2m+1.
\eeq
For the positive choice of $\Phi$, we have
\beq
\omega=e^{im\pi/N},
\quad
\omega^N=(-1)^m.
\eeq
As explained above, the negative choice of $\Phi$ follows by complex conjugation. We separate two subcases according to whether $m$ is odd or even.

\subsubsection{Case \texorpdfstring{$n=4m+2$}{n=4m+2}, \texorpdfstring{$m$}{m} odd}

Choose $\gamma=0$ and $c=N/2$ in \cref{betaj}. Then
\beq\label{newbetaj}
\beta_j=(-1)^{j+1}j+\frac N2.
\eeq
Write again
\beq
P(t)=T_1(t)\cdots T_n(t),
\quad
T_j(t)=
\begin{pmatrix}
1 & t\omega^{\beta_j}\\
t\omega^{-\beta_j} & 1
\end{pmatrix}.
\eeq
For every $j=1,\dots,N$,
\beq
\beta_{j+N}=-\beta_j+(1+(-1)^j)N.
\eeq
Since $m$ is odd, $\omega^N=-1$, and hence
\beq
\omega^{\beta_{j+N}}=\omega^{-\beta_j}.
\eeq
Therefore, with
\beq
J:=\begin{pmatrix}0&1\\1&0\end{pmatrix},
\eeq
we have
\beq
T_{j+N}(t)=JT_j(t)J,
\quad 1\le j\le N.
\eeq
Defining
\beq
Q(t):=T_1(t)\cdots T_N(t),
\eeq
we obtain
\beq
P(t)=Q(t)JQ(t)J.
\eeq

Now consider
\beq
Q(t)=T_1(t)\cdots T_m(t)T_{m+1}(t)\cdots T_{2m}(t)T_N(t).
\eeq
For $1\le j\le m$, we check from \cref{newbetaj} that
\beq
\beta_j-\beta_{2m-j+1}=\pm N.
\eeq
Since $m$ is odd, $\omega^{\pm N}=-1$, and hence
\beq
T_j(t)T_{2m-j+1}(t)=(1-t^2)I.
\eeq
Pairing factors from the middle outward gives
\beq
Q(t)=(1-t^2)^mT_N(t).
\eeq
Moreover,
\beq
\beta_N=\frac{3N}{2},
\quad
\omega^{\beta_N}=\exp \left(\frac{3im\pi}{2}\right)
=
(-1)^{(m+1)/2}i .
\eeq
Hence
\beq
T_N(t)=
\begin{pmatrix}
1&(-1)^{(m+1)/2}it\\
-(-1)^{(m+1)/2}it&1
\end{pmatrix}.
\eeq
A direct multiplication gives
\beq
T_N(t)JT_N(t)J=(1-t^2)I.
\eeq
Therefore
\beq
\begin{aligned}
P(t)
&=
(1-t^2)^{2m}T_N(t)JT_N(t)J
=
(1-t^2)^{2m+1}I\\
&= (1-t^2)^N I.
\end{aligned}
\eeq
Thus
\beq
\lambda_{2k}=(-1)^k\binom{N}{k}=(-1)^k\binom{n/2}{k},
\eeq
and the preceding reduction gives $A_0^{(k)}=\lambda_{2k}$. This proves \cref{th:eq_39} for $n=4m+2$ with $m$ odd.

\subsubsection{Case \texorpdfstring{$n=4m+2$}{n=4m+2}, \texorpdfstring{$m$}{m} even}

Again set
\beq
N:=\frac n2=2m+1.
\eeq
Choose
\beq
\gamma=\frac{N}{2m},
\quad
c=0
\eeq
in \cref{betaj}. Once again write
\beq
P(t)=T_1(t)\cdots T_n(t),
\quad
T_j(t)=
\begin{pmatrix}
1 & t\omega^{\beta_j}\\
t\omega^{-\beta_j} & 1
\end{pmatrix}.
\eeq
For every $j=1,\dots,N$, direct calculation gives
\beq
\beta_{j+N}=-\beta_j\pm N.
\eeq
Since $m$ is even, $\omega^N=1$, and hence
\beq
\omega^{\beta_{j+N}}=\omega^{-\beta_j}.
\eeq
Therefore
\beq
T_{j+N}(t)=JT_j(t)J,
\quad 1\le j\le N,
\eeq
and, with $Q(t):=T_1(t)\cdots T_N(t)$,
\beq
P(t)=Q(t)JQ(t)J.
\eeq

Now focus on
\beq
Q(t)=T_1(t)\cdots T_m(t)T_{m+1}(t)\cdots T_{2m}(t)T_N(t).
\eeq
For $1\le j\le m$, the definition of $\beta_j$ gives
\beq
\beta_j-\beta_{N-j}
=
(-1)^{j+1}\left(N+\frac Nm\right)
=
(-1)^{j+1}N\left(\frac{m+1}{m}\right).
\eeq
Since $m+1$ is odd,
\beq
\omega^{\beta_j-\beta_{N-j}}=-1.
\eeq
Hence
\beq
T_j(t)T_{N-j}(t)=(1-t^2)I.
\eeq
Pairing factors from the middle outward gives
\beq
Q(t)=(1-t^2)^mT_N(t).
\eeq
Finally,
\beq
\beta_N=\frac{N^2}{2m},
\quad
\omega^{\beta_N}=\exp \left(\frac{iN\pi}{2}\right)
=
i,
\eeq
where we used $N=2m+1\equiv 1 \pmod 4$, since $m$ is even. Thus
\beq
T_N(t)=
\begin{pmatrix}
1&it\\
-it&1
\end{pmatrix},
\eeq
and again
\beq
T_N(t)JT_N(t)J=(1-t^2)I.
\eeq
Therefore
\beq
\begin{aligned}
P(t)
&=
(1-t^2)^{2m}T_N(t)JT_N(t)J
=
(1-t^2)^{2m+1}I\\
&=
(1-t^2)^N I.
\end{aligned}
\eeq
It follows that
\beq
\lambda_{2k}=(-1)^k\binom{N}{k}=(-1)^k\binom{n/2}{k}.
\eeq
By the preceding reduction, $A_0^{(k)}=\lambda_{2k}$. This completes the proof of \cref{th:eq_39}.

\section{Conclusion}\label{sec:conclusion}

We have given a rigorous proof of the even-$n$ universally robust dynamical decoupling sequences introduced in Ref.~\cite{PhysRevLett.118.133202}. Working in the same effective pulse-cycle model and using the error variable $\epsilon=\sqrt{1-p}$, we proved that the UR phase prescription cancels all nonconstant Taylor coefficients of the normalized HS overlap
$G(\epsilon)=\frac12\Tr(U_0^\dagger U_\epsilon)$
through order $\epsilon^{n-1}$. Since only even powers of $\epsilon$ occur, this is equivalent to cancellation through order $\epsilon^{n-2}$. Therefore
\beq
G(\epsilon)=1+O(\epsilon^n),
\quad
1-F=O(\epsilon^n).
\eeq
Moreover, the order-$\epsilon^n$ coefficient of the fidelity error is not identically zero as a function of $\alpha$. Since $1-p=\epsilon^2$, this is precisely the $(1-p)^{n/2}$ scaling claimed in Ref.~\cite{PhysRevLett.118.133202}.

Our proof also makes the cancellation mechanism explicit. The Taylor expansion of the HS overlap reduces the problem to trace identities involving the operators $\tilde P_{n-2k}$. These traces are then organized by the alternating linear form $\mathcal L_r$, the signature $S(r)$, and the weighted signature $W(r)$. In this way, the high-order robustness of the UR construction is traced to a combinatorial and root-of-unity structure, and the vanishing of the relevant amplitude coefficients is reduced to Fourier-type identities.

The proof applies to even-length UR sequences within the effective model of Ref.~\cite{PhysRevLett.118.133202}, under the same assumptions used there: coherent evolution over a single DD sequence, identical pulses within that sequence apart from the controlled phase shifts, and an environment whose correlation time exceeds the sequence duration.

Several directions remain open. It would be interesting to extend the analysis to odd-$n$ constructions, to related robust DD and composite-pulse families, and to models with additional control constraints or faster environmental fluctuations. More broadly, the present work shows that the URDD construction is governed by a particular algebraic and Fourier structure. We believe that this structure may be useful in the systematic design and analysis of new robust control sequences in different and more general contexts.

\section*{Acknowledgment}
This material is based upon work supported in part by the U.S. Army Research Laboratory and the U.S. Army Research Office under contract/grant number W911NF2310255. D.~D'Alessandro also acknowledges support from a Scott Hanna Professorship at Iowa State University.

\appendix

\section{\texorpdfstring{Proof of \cref{eq:Altrecurs}}{Proof of Eq. Altrecurs}}
\label{app:Altrecurs}

Let
\beq
B_j:=(x_{r_j+1},\dots,x_{r_{j+1}}),
\quad j=0,1,\dots,2k,
\eeq
where $r_0=0$ and $r_{2k+1}=n$.  Thus the right-hand side of
\cref{eq:Altrecurs} is obtained by taking the alternating sum on each block
$B_j$, and multiplying that block by the prefactor $(-1)^{j+r_j}$.

We compare the coefficients of the variables $x_1,\dots,x_n$ on both sides.
Fix $j\in\{0,\dots,2k\}$, and let $i$ satisfy $r_j<i\le r_{j+1}$. Then
$x_i$ appears only in the $j$'th block, and its coefficient on the
right-hand side of \cref{eq:Altrecurs} is
\beq
(-1)^{j+r_j}(-1)^{ i-(r_j+1)}
=
(-1)^{i+j-1}.
\eeq
Hence, within each block $B_j$, the signs alternate from one entry to the
next.

Next we check what happens at a breakpoint $r_j$ with $1\le j\le 2k$.
The last entry of the previous block is $x_{r_j}$, and its coefficient is
\beq
(-1)^{j-1+r_{j-1}}(-1)^{ r_j-r_{j-1}-1}
=
(-1)^{j+r_j}.
\eeq
The first entry of the next block is $x_{r_j+1}$, and its coefficient is
\beq
(-1)^{j+r_j}(-1)^0
=
(-1)^{j+r_j}.
\eeq
Therefore the sign is repeated when we pass from $x_{r_j}$ to
$x_{r_j+1}$, exactly as required by the defining description of
$\mathcal L_r(X)$.

Finally, on the first block $B_0=(x_1,\dots,x_{r_1})$, the coefficient of
$x_1$ is
\beq
(-1)^{0+r_0}(-1)^0=1,
\eeq
since $r_0=0$. Thus the sign pattern produced by the right-hand side starts
with $+1$, alternates within each block, and repeats the sign at each
breakpoint $r_j$. This is precisely the original definition of
$\mathcal L_r(X)$.

Hence the block formula \cref{eq:Altrecurs} is identical to the original
sign-pattern definition of $\mathcal L_r(X)$.

If we allow $r_j=r_{j+1}$ for some $j$, then the corresponding block is
empty and contributes $0$ by convention, so the same argument remains valid.

\section{\texorpdfstring{Proof of \cref{Nicelemma}}{Proof of Lemma Nicelemma}}
\label{app:Nicelemma}

\begin{proof}
Applying \cref{eq:Altrecurs} to $X_N=(1,2,\dots,n)$, we obtain
\beq\label{eq:insert6}
\mathcal{L}_r(X_N)=\sum_{j=0}^{2k} (-1)^{j+r_j} \mathcal{L}(r_j+1,r_j+2,\dots,r_{j+1}).
\eeq
We also have
\beq
\begin{aligned}
    &\mathcal{L}(r_j+1,\dots,r_{j+1})=r_j\sum_{s=1}^{r_{j+1}-r_j} (-1)^{s-1}+\sum_{s=1}^{r_{j+1}-r_j}(-1)^{s-1}s\\
    &\quad=r_j\frac{(1-(-1)^{r_{j+1}-r_j})}{2}\\
    &\quad\quad+(-1)^{r_{j+1}-r_j} \left( \frac{(-1)^{r_{j+1}-r_j}-1-2(r_{j+1}-r_j)}{4}\right)\\
&\quad=\frac{1-(-1)^{r_{j+1}-r_j}}{4}+\frac{r_j-r_{j+1}(-1)^{r_{j+1}-r_j}}{2},
\end{aligned}
\eeq
where in the last equality we used the fact that $\sum_{s=1}^a (-1)^{s-1}s$ is $-\frac{a}{2}$ if $a$ is even and $\frac{a+1}{2}$ if $a$ is odd. This identity is easily proved by induction on $a$, and we use it in the form
\beq
\sum_{s=1}^a (-1)^{s-1}s=(-1)^a \left( \frac{(-1)^a-1-2a}{4}\right)
\eeq
with $a=r_{j+1}-r_j$.

Using this in \cref{eq:insert6}, we have
\beq
\begin{aligned}
\mathcal{L}_r(X_N)
&=
\frac14
\sum_{j=0}^{2k} (-1)^j
\left((-1)^{r_j}-(-1)^{r_{j+1}}\right)\\
&\quad+
\frac12
\sum_{j=0}^{2k}(-1)^j
\left((-1)^{r_j}r_j-(-1)^{r_{j+1}}r_{j+1}\right).
\end{aligned}
\eeq
Using $r_0=0$, $r_{2k+1}=n$, and $n$ even, the first sum telescopes to
\beq
\frac12\sum_{j=1}^{2k}(-1)^{j+r_j}
=
\frac{S(r)}{2}.
\eeq
The second sum telescopes to
\beq
\sum_{j=1}^{2k}(-1)^{j+r_j}r_j-\frac n2
=
W(r)-\frac n2.
\eeq
Therefore
\beq
\mathcal L_r(X_N)
=
\frac{S(r)}2+W(r)-\frac n2
=
\frac{S(r)-n}{2}+W(r),
\eeq
which is \cref{eq:nicelemmaeq}.
\end{proof}

%


\end{document}